\documentclass[letterpaper, 10 pt, conference]{ieeeconf}  

\IEEEoverridecommandlockouts                              
\usepackage{amsfonts,amsmath,amssymb,booktabs,color,float,fullpage,graphicx}
\usepackage[mathscr]{eucal}
\usepackage{xfrac}
\usepackage{comment}
\usepackage{algorithm,algpseudocode}
\usepackage[nolabel]{showlabels}
\usepackage{hyperref}
\usepackage{cleveref}
\graphicspath{{./}{./Graphics/}}
\newtheorem{proposition}{Proposition}

\newtheorem{example}{Example}

\newcommand{\eig}{\mathfrak{e}}
\newcommand{\R}{\mathbb{R}}
\newcommand{\C}{\mathbb{C}}
\newcommand{\Koopman}{\mathcal{K}}
\newcommand{\data}[1]{\mathbf{x}(t_{#1})}
\newcommand{\interval}[1]{\Delta t_{#1}}
\newcommand{\eigvec}{\mathbf{a}}

\newcommand{\define}{\overset{\text{def}}{=}}
\newcommand{\matC}{\mathbf{C}(\alpha,\beta)}
\newcommand{\VF}{\mathsf{V}}
\newcommand{\Flow}{\Phi}

\newcommand{\sx}{\mathbf{x}}
\newcommand{\Loss}{\mathcal{L}}
\newcommand{\bfa}{\mathbf{a}}
\newcommand{\bfG}{\mathbf{G}}
\newcommand{\bfC}{\mathbf{C}}
\newcommand{\lb}{\left\{}
\newcommand{\rb}{\right\}}

\title{\LARGE \bf
{K}oopman Representations for Non-Vanishing Time Intervals: \\ An Optimization Approach and Sampling Effects}

\author{Younghwan Cho and Richard Sowers
\thanks{Younghwan Cho is with the Department of Industrial and Enterprise Systems Engineering, University of Illinois, Urbana, 61801, IL, United States
        {\tt\small yc71@illinois.edu}}%
\thanks{Richard Sowers is with the Department of Mathematics and the Department of Industrial and Enterprise Systems Engineering, University of Illinois, Urbana, 61801, IL, United States {\tt\small r-sowers@illinois.edu}}%
\thanks{This work has been submitted to the IEEE for possible publication. Copyright may be transferred without notice, after which this version may no longer be accessible.}
}
\begin{document}

\maketitle
\thispagestyle{empty}
\pagestyle{empty}

\begin{abstract}
Koopman operator theory is a key tool in data assimilation of complex dynamical systems, with the potential to be applied to multimodal data. We formulate the problem of learning Koopman eigenfunctions from observations at arbitrary, possibly non-vanishing, time intervals as an optimization problem. Analysis of the formulation reveals aliasing induced by oscillatory dynamics and the sampling pattern, making an inherent identifiability limit explicit. The analysis also uncovers phase alignment near the true Koopman frequency, which creates a steep loss valley and demands careful optimization. We further show that irregular sampling can break aliasing and lead to phase cancellation. Numerical results demonstrate the efficacy of the proposed method under large regular time intervals compared to generator extended dynamic mode decomposition, and support the idea that irregular sampling can help recover the true Koopman spectrum.
\end{abstract}

\section{Introduction}
The seminal idea of Koopman~\cite{koopman1931hamiltonian} maps flows of possibly nonlinear dynamics to a linear operator on a function space without loss of information; the tradeoff is that, generically, this space is infinite-dimensional. Under certain regularity conditions, the linear evolution of mapped flows admits a spectral decomposition for a broad family of dynamics~\cite{koopman1931hamiltonian}, motivating significant interest in applying Koopman decompositions to complex dynamical systems~\cite{mezic2013analysis,brunton2016extracting,kou2022data,avila2020data}. Since the Koopman operator is infinite-dimensional, finite-dimensional approximations are of great importance. However, an arbitrary finite-dimensional subspace is not necessarily Koopman-invariant. A lack of invariance can produce spurious spectral objects and only locally accurate predictions~\cite{mezic2020spectrum, bevanda2021koopman}. To address the aforementioned, recent studies~\cite{folkestad2020extended,korda2020optimal,kaiser2021data, deka2023path} work with subspaces spanned by Koopman eigenfunctions, which are Koopman-invariant by construction, rather than arbitrary dictionaries of functions~\cite{williams2015edmd, klus2020data}.


Our interest, broadly, is applying Koopman decompositions to \emph{multimodal} data. For example, sensor data of human gait~\cite{kaur2023deep} contains multiple gait events with temporal irregularity, e.g., left and right heel strike are not exactly periodic due to subjects' heterogeneity or the morbidity of neurodegenerative disorder~\cite{martin2006gait}. Loop detectors detect traffic at fixed but potentially irregular locations~\cite{fujito2006effect}.  As such, we have various interesting phenomena hidden in the complex data, and our particular interest is to learn Koopman eigenfunctions with observations of irregular time intervals. Several studies~\cite{korda2020optimal, cho2025koopman} address Koopman eigenfunction learning under irregular sampling. In \cite{korda2020optimal}, local Koopman eigenfunctions are constructed on non-recurrent sets, and these functions are extended to the rest of the state space via interpolation. On the other hand, \cite{cho2025koopman} searches for eigenfunctions directly over the entire space by minimizing the Koopman invariance residual, and its robustness has been demonstrated empirically under different sampling scenarios. 

In this paper, we propose an extension of \cite{cho2025koopman} to complex Koopman spectral objects. We formulate the problem of learning Koopman eigenfunction as an optimization task that operates directly on state observations at arbitrary times.
Analysis of the loss landscape reveals that (1) the formulation explicitly reflects aliasing induced by the sampling pattern and that (2) phase alignment near the true frequency creates a steep valley. The identifiability limit is not directly exposed in convex formulations such as EDMD~\cite{williams2015edmd}, which may nonetheless be affected by the same information loss. Third, we show analytically and numerically that irregular sampling can break aliasing and improve the identifiability of Koopman spectral objects. We expect that our formulation and its analysis provide general insights to optimization of Koopman-based methods.
\section{Preliminaries}
Koopman operator theory allows us to analyze continuous-time flows on some Euclidean space \(\R^D\).
With a Lipschitz continuous or continuously differentiable vector field $\VF:\R^D\to \R^D$, first define the flow map $\{\Flow_t\}_{t\in \R}$ as
the solution of 
\begin{equation}\label{equ:flow} \Phi_t(\sx)=\sx+\int_{s=0}^t \VF(\Flow_s(\sx))ds \qquad t\in \R. \end{equation} For arbitrary map (or \lq\lq observable") \(f:\R^D\to \C\), one can define the Koopman group \(\{\Koopman_t\}_{t\in \R}\) as follows:
\begin{equation*} \left(\Koopman_t f\right)(\sx) \define f(\Flow_t(\sx))=(f\circ \Flow_t)(\sx) \qquad \sx\in \R^D. \end{equation*}
The main result of Koopman \cite{koopman1931hamiltonian} was that in certain settings (restriction of \(\Koopman_t\) to $L^2(\R^D)$ and unitary {$\Koopman_t$}), we can construct spectral decompositions of $\{\Koopman_t\}_{t\in \R}$.  In other words, there exist pairs $(\eig,\lambda)$ of functions $\eig :\R^D \to \C$ (eigenfunctions) and $\lambda\in \C$ (eigenvalues) such that \( \left(\Koopman_t \eig\right)(\sx) = e^{\lambda t}\eig(\sx),\enspace \sx\in \R^D\) or, equivalently,
\begin{equation} \label{equ:almostKoopman} \eig(\Phi_t(\sx))=e^{\lambda t}\eig(\sx)\qquad \sx\in \R^D\end{equation}
The Koopman decomposition states that, under the correct hypothesis, the collection of eigenfunctions spans the collection of all square-integrable functions in a proper sense. 
The decomposition provides insights into nonlinear dynamics, including reduced order modeling~\cite{mezic2004comparison} and invariant and periodic structures of the state space \cite{budivsic2009approximate}.

\section{Problem formulation}
Suppose that we want to approximate a Koopman decomposition using data $\{\data{n}\}_{n=1}^N\subset \R^D$. If we would have a Koopman eigenpair $(\eig,\lambda)$, we would have
\begin{equation*} \eig(\data{n+1})=e^{\lambda (t_{n+1}-t_n)}\eig(\data{n}) \end{equation*}
for all $n\in \{1,2\dots N\}$, or, equivalently,
\begin{equation}\label{eq:basicloss} \frac{1}{N}\sum_{n=1}^N \left|\eig(\data{n+1})-e^{\lambda (t_{n+1}-t_n)}\eig(\data{n})\right|^2 =0. \end{equation}
More generally, if $\phi:\R^N\to \C$, define
\begin{equation*} \Loss(\phi,\lambda)\define \frac{1}{N}\sum_{n=1}^N\left|\phi(\data{n+1})-e^{\lambda (t_{n+1}-t_n)}\phi(\data{n})\right|_\C^2 \end{equation*}
(where $|\cdot|_\C$ is modulus in $\C$)
then $\Loss(\eig,\lambda)=0$ if $(\eig,\lambda)$ is a Koopman eigenpair.  Conversely, if $\Loss(\phi,\lambda)=0$ and there are \lq\lq enough" $t_n$'s, $(\phi,\lambda)$ should be \lq\lq close" to an eigenpair (we will not make this precise).

Numerically, we might fix a space $V$ of maps from $\R^D$ to $\C$ and minimize
\begin{equation*} \min_{\substack{\phi\in V \\ \lambda\in \C}}\Loss(\phi,\lambda). \end{equation*}
Given theoretical knowledge that there is a Koopman decomposition, this minimization problem should be close to zero if $V$ is \lq\lq rich" enough.

Fix some (linearly independent) collection functions $\{g_m\}_{m=1}^M$ , each of the $g_m$'s being a map from $\R^D$ to $\C$. Define $\bfG:\R^D\to \C^M$  as
\begin{equation*} \bfG(x)\define \begin{pmatrix} g_1(x) \\ g_2(x) \\ \vdots  \\ g_M(x)\end{pmatrix} \end{equation*}
for all $x\in \R^D$.  The span of $\{g_1,g_2\dots g_M\}$ is then
\begin{equation*} \lb \bfa^\dagger \bfG: \bfa\in \C^M \rb \end{equation*} where \(\dagger\) is a conjugate transpose.
We are then interested in the minimization problem
\begin{equation*} \min_{\substack{\bfa\in \C^M \\ \lambda\in \C}}\Loss(\bfa^\dagger\bfG,\lambda) \end{equation*}

Let us write out $\Loss(\bfa^\dagger\bfG,\lambda)$. Expanding the modulus squared and averaging over $n$, we obtain 
\begin{equation*} \Loss(\bfa^\dagger\bfG,\lambda)
= \bfa^\dagger \bfC(\lambda) \bfa \end{equation*}
where
\begin{equation} \label{eq:bfCDef}
\begin{aligned} \bfC(\lambda) &\define
\frac{1}{N}\sum_{n=1}^N \left(\bfG(\data{n+1})- e^{\lambda(t_{n+1}-t_n)}\bfG(\data{n})\right)\\
&\quad \qquad \times \left(\bfG(\data{n+1})-e^{\lambda(t_{n+1}-t_n)} \bfG(\data{n})\right)^\dagger. \end{aligned}\end{equation}

Write time intervals as \(\interval{n} = t_{n+1} - t_{n}\), and \(\lambda = \alpha + i\beta \in \C\).
We can now solve the optimization problem to find the approximation of \((\eig,\lambda)\) as follows:
\begin{equation}\label{eq:min_prob}
    \min_{\alpha,\beta \in \R,\, 
    \eigvec \in \mathbb{C}^M, \,
    \|\eigvec\|_2^2 = 1} \Loss\left(\alpha,\beta, \bfa \right).
\end{equation}
First, we deduce useful basic properties of \eqref{eq:min_prob}.  In \eqref{eq:min_prob}, one can show that \(\matC\) is Hermitian positive semi-definite. Moreover, \(\matC\) is analytic function of \(\alpha,\beta\), and the loss of \eqref{eq:min_prob} is symmetric w.r.t. \(\beta\). Given some \(\alpha,\beta\), the optimal solution of \eqref{eq:min_prob} is the smallest eigenvalue-eigenvector pair of \(\matC\). Since \(\matC\) is Hermitian positive semi-definite, the constrained minimum over \(\bfa\) is obtained at the eigenvector corresponding to the smallest eigenvalue of \(\matC\).

We have two key linear operators in \eqref{eq:min_prob}: the Koopman operator and the \(\matC\) depending on data and the Koopman eigenvalue, whose spectra determine the regularity of \eqref{eq:min_prob}. We avoid assumptions on the Koopman spectrum, as the underlying dynamics are often inaccessible.
Without further assumptions, \eqref{eq:min_prob} learns/approximates a point spectrum of an unknown Koopman operator. \(\mathbf{C}(\alpha, \beta)\) has a real and non-negative spectrum. We will assume that the spectrum of \(\matC\) is simple and support this with that the occurrence of non-simple Hermitian matrices is non-generic~\cite{tao2023topics}. Assuming the simple spectrum \(\matC\) gives us a better regularity on the eigenvalue map, i.e., a map from a Hermitian matrix to its smallest eigenvalue. 
This will enable us to design an optimization scheme based on the first-order information of the loss function. Denote by $\lambda_M: H(M) \rightarrow \C$ the smallest-eigenvalue map and by $u_M: H(M) \rightarrow \C^M$ the associated unit-eigenvector map on the space of $M$-dimensional Hermitian matrices. Then, we can write \eqref{eq:min_prob} as the following form:
\begin{equation}\label{eq:eigenproblem}
    \underset{\substack{
  \alpha,\beta \in \R}}{\min} \lambda_M\left( \matC\right).
\end{equation}
For brevity, we write \(\lambda_M(\alpha,\beta)\) for \(\lambda_M(\matC)\), and similarly for \(u_M\). We have the following result that can be directly applied to \eqref{eq:eigenproblem}.
\begin{proposition}\label{prop:simple-spectrum-gradient}
Assume $\matC\in\mathbb{C}^{M\times M}$ is \(C^1\) for all $(\alpha,\beta)$.
Suppose that at $(\alpha_0,\beta_0)$ the eigenvalue $\lambda_M(\alpha_0,\beta_0)$ is simple, i.e., $\mathbf{C}(\alpha_0,\beta_0)$ has no repeated eigenvalues. 
Let $u_M(\alpha_0,\beta_0)\in\mathbb{C}^M$ be a unit eigenvector associated with $\lambda_M(\alpha_0,\beta_0)$, i.e.
\begin{align*}
\mathbf{C}(\alpha_0,\beta_0)u_M(\alpha_0,\beta_0)&=\lambda_M(\alpha_0,\beta_0)\,u_M(\alpha_0,\beta_0), \\
u_M(\alpha_0,\beta_0)^\dagger u_M(\alpha_0,\beta_0)&=1.    
\end{align*}
Then $\lambda_M$ is differentiable at $(\alpha_0,\beta_0)$ and its gradient with respect to $(\alpha,\beta)$ satisfies the following:
\begin{equation} \label{eq:eigval-gradient}\begin{aligned}
&\begin{bmatrix}
\displaystyle \frac{\partial}{\partial\alpha}\lambda_M(\alpha_0,\beta_0)\\[6pt]
\displaystyle \frac{\partial}{\partial\beta}\lambda_M(\alpha_0,\beta_0)
\end{bmatrix}\\
&\qquad =
\begin{bmatrix}
u_M(\alpha_0,\beta_0)^\dagger \left(\frac{\partial}{\partial\alpha}\mathbf{C}(\alpha_0,\beta_0)\right)u_M(\alpha_0,\beta_0)\\[6pt]
u_M(\alpha_0,\beta_0)^\dagger \left(\frac{\partial}{\partial\beta}\mathbf{C}(\alpha_0,\beta_0)\right)u_M(\alpha_0,\beta_0)
\end{bmatrix}.
\end{aligned}
\end{equation}
\end{proposition}
\begin{proof}
Using the one-parameter perturbation result of the eigenvalue \((\alpha_0,\beta_0)\) and the associated eigenvector in \cite{tao2023topics} or \cite{rellich1969perturbation}, we take derivatives of \(\mathbf{C}(\alpha_0,\beta_0)u_M(\alpha_0,\beta_0)=\lambda_M(\alpha_0,\beta_0)\,u_M(\alpha_0,\beta_0)\) for both sides. The results follows by using \(u_M(\alpha_0,\beta_0)^\dagger u_M(\alpha_0,\beta_0)=1\) to organize the terms.
\end{proof}
\section{Loss landscape analysis}
Having established differentiability, we analyze the structure of the loss landscape from global to local properties to understand the optimization problem. 
By regular sampling, we refer to periodic sampling with a fixed, possibly non-vanishing sampling time interval. Our first observation is that the loss function of \eqref{eq:min_prob} exhibits a periodic structure in the \(\beta\)-axis under regular sampling.
\begin{proposition}\label{prop:periodicity}
Assume \( \interval{n} = \delta > 0, \enspace \forall n \in [N].\)
Fix a feasible coefficient vector \(\eigvec\), and define
\begin{align*}
\mathcal{L}(\alpha,\beta) &\define \Loss(\alpha,\beta,\bfa)\\
\phi_n &\define \eigvec^\dagger \bfG(\data{n}), \\
\psi_n &\define e^{\alpha\delta}\phi_n \bar{\phi}_{n+1}.    
\end{align*}
Then, for every fixed \(\alpha\), the map \(\beta \mapsto \mathcal{L}(\alpha,\beta)\) is periodic with period \(2\pi/\delta\):
\begin{equation*}
\mathcal{L}(\alpha,\beta+\sfrac{2\pi}{\delta})
=
\mathcal{L}(\alpha,\beta),
\qquad \forall \beta \in \R.    
\end{equation*}
\end{proposition}
\begin{proof}
By definition, the loss can be written as
\begin{equation*}
    \mathcal{L}(\alpha,\beta)=
    \frac{1}{N}\sum_{n=1}^N
    \left|
    \phi_{n+1}-e^{(\alpha+i\beta)\interval{n}}\phi_n
    \right|^2.
\end{equation*}
Expanding \eqref{eq:bfCDef}, we have
\vspace{-4pt}
\begin{align*} \bfC(\lambda)
&= \frac{1}{N}\sum_{n=1}^N \bfG(\data{n+1})\bfG^T(\data{n+1})\\
&\quad + \frac{1}{N}\sum_{n=1}^N e^{2\Re(\lambda)\interval{n}} \bfG(\data{n})\bfG^T(\data{n})\\
&\quad - \frac{2}{N}\sum_{n=1}^N \Re\lb e^{\lambda\interval{n}}\rb \\
&\quad \qquad \times \bfG(\data{n+1})\bfG^T(\data{n})
\end{align*}
We multiply \(\bfa^\dagger\) on the left and $\bfa$ on the right. Using the definition of \(\psi_n\), \(\phi_n\) and 
\(\interval{n}=\delta\), we have
\begin{align*}
    &\mathcal{L}(\alpha,\beta)\\
    &=
    \frac{1}{N}\sum_{n=1}^N
    \Big(
    |\phi_{n+1}|^2
    +
    |e^{(\alpha+i\beta)\delta}\phi_n|^2\\
    &-
    2\Re\!\big(
        \phi_{n+1}\overline{e^{(\alpha+i\beta)\delta}\phi_n}
    \big)
    \Big) \\
    &=
    \frac{1}{N}\sum_{n=1}^N
    \left(
    |\phi_{n+1}|^2
    +
    e^{2\alpha\delta}|\phi_n|^2
    -
    2\Re\!\left(
        e^{i\beta\delta} e^{\alpha\delta}\phi_n\bar{\phi}_{n+1}
    \right)
    \right) \\
    &=
    \frac{1}{N}\sum_{n=1}^N
    \left(
    |\phi_{n+1}|^2
    +
    e^{2\alpha\delta}|\phi_n|^2
    \right)
    -
    \frac{2}{N}\Re\!\left(
        \sum_{n=1}^N e^{i\beta\delta}\psi_n
    \right).
\end{align*}
Then, for any \(\beta \in \R\), it is clear that \(\mathcal{L}\!(\alpha,\beta+\sfrac{2\pi}{\delta}) = \mathcal{L}(\alpha,\beta)\). Therefore, \(\mathcal{L}(\alpha,\beta)\) is periodic in \(\beta\) with period \(2\pi/\delta\).
\end{proof}
The result shows that coarse sampling induces stronger aliasing in the frequency direction. More precisely, predicted frequencies separated by integer multiples of \(2\pi/\delta\) are indistinguishable to the loss. In contrast, finer sampling increases the period in the \(\beta\)-direction and makes the optimization landscape less affected by phase aliasing. This interpretation is consistent with \emph{the Nyquist--Shannon sampling theorem}. That is, for exact identification of a signal with a frequency band of \([-\beta_M,\beta_M]\), the sampling rate needs to satisfy \(\frac{2\pi}{\delta} > 2\beta_M\)~\cite{shannon2006communication}.  In this study, we do not assume a frequency band for the true dynamics, as we are machine-learning the system of interest. Hence, the smaller \(\delta\) is, the more amenable our problem is to numerical methods. In practice, expert knowledge can be used to propose a plausible bandwidth a priori to partially resolve the identification issue.

Using \Cref{prop:simple-spectrum-gradient}, the partial derivatives of 
\(\lambda_M\left(\alpha,\beta\right)\) are
\begin{equation}\label{eq:gradient}
\begin{split}
    \frac{\partial\lambda_M\left(\alpha,\beta\right)}{\partial \alpha} &= \frac{2}{N}\sum_{n=1}^N \interval{n}e^{\alpha\interval{n}}\Big( e^{\alpha\interval{n}} |\phi_n|^2 \\
    &- \Re{(e^{i\beta\interval{n}}\phi_n \bar{\phi}_{n+1})}  \Big). \\
    \frac{\partial\lambda_M\left(\alpha,\beta\right)}{\partial \beta} &= \frac{2}{N}\sum_{n=1}^N  \Im{(\interval{n}\psi_n e^{i\beta \interval{n}})}. 
\end{split}
\end{equation}
Expanding $\psi_n$ and writing each \(\phi_n\) in polar form,
$\phi_n = |\phi_n|\,e^{i\arg\phi_n}$, we obtain
\begin{equation*}
\begin{split}
&\frac{\partial\lambda_M(
\alpha,\beta)}{\partial \beta}\\
&= \frac{2}{N}\sum_{n=1}^{N}
   \Im\!\Big(\interval{n}\,|\psi_n|\,
   e^{i(\beta\interval{n} + \arg\phi_n - \arg\phi_{n+1})}\Big)\\[6pt]
&= \frac{2}{N}\sum_{n=1}^{N}
   \Im\!\Big(\interval{n}\,e^{(\alpha+\alpha^*)\interval{n}}\,
   |\phi_n|^{2}\,
   e^{i(\beta-\beta^*)\interval{n} + i2\pi k_n}\Big)\\[6pt]
&= \frac{2}{N}\sum_{n=1}^{N}
   \interval{n}\,e^{(\alpha+\alpha^*)\interval{n}}\,|\phi_n|^{2}\,
   \sin\!\left((\beta-\beta^*)\interval{n}\right),
\end{split}
\end{equation*}
where $\beta^*$ denotes the true Koopman frequency. Koopman invariance implies $\phi_{n+1} = e^{(\alpha^*+i\beta^*)\interval{n}}\phi_n$ and it follows that $\arg\phi_{n+1}= \arg\phi_n + \beta^*\interval{n} + 2\pi k_n$, where $k_n\in\mathbb{Z}$ is the unique integer such that $\arg\phi_{n+1}\in(-\pi,\pi]$. 

Let us denote the prediction error as $\Delta \beta = \beta-\beta^*$. When $\Delta \beta \approx0$ we have the approximation
\[
\sin\!\left(\Delta \beta\interval{n}\right) \approx \Delta \beta\interval{n},
\] so the summands become exponential and quadratic in \(\interval{n}\), and their phases align, resulting in a large slope of the loss function. Thus, locally, non-vanishing time intervals make the magnitude of the gradient large, requiring a careful step-size selection near an optimum. Also, the magnitude of the predicted frequency $\beta$ does not influence the local slope near the optimal solution, but is instead determined by the absolute error $\Delta \beta$. When the prediction error is large, the summand becomes highly oscillatory as a function of $\interval{n}$, leading to phase cancellation across samples. However, this cancellation does not occur under regular sampling.

Under \emph{non-vanishing regular sampling}, the loss landscape exhibits two distinct effects. Globally, the periodicity in \(\beta\) creates aliasing. The aliasing reflects a natural identifiability limit imposed by the sampling pattern. This ambiguity is not observed in the widely-adopted formulation of the EDMD family that solves a least-square problem. Since EDMD assumes a fixed time step, the same information loss is present, but potentially manifests as bias in the approximated operator rather than as periodicity in a loss landscape. Locally, phase alignment induces a steep valley near the true \(\beta\), making optimization sensitive to step-size selection. Irregular sampling also produces this local geometry. This is consistent with the observation from \cite{cho2025koopman} that a high precision line search method is effective in learning Koopman eigenvalues. 
    
\emph{Irregular sampling} have additional effects. First, temporal irregularity breaks the periodicity in \Cref{prop:periodicity} responsible for aliasing. This implies that we have better identifiability under irregular sampling times. This raises the possibility of learning high-frequency signals using a sampling rate lower than the one implied by the Shannon sampling theorem via irregular sampling, as in compressed sensing~\cite{candes2008introduction}.
Secondly, irregular sampling disrupts phase alignment in computing partial derivatives of \(\beta\) when \(\Delta \beta\) is non-vanishing. Hence, we would expect a gentler slope of loss landscape when our predictive error in \(\beta\) is large. A mixture of these effects will construct a loss landscape of the problem. We leave further analysis to future work.
\section{Numerical Results}
We test the proposed formulation under large time intervals and under high-frequency dynamics. The latter is designed to demonstrate the benefits of irregular sampling. In both experiments, we compare the spectrum learned from our new formulation with the one from gEDMD~\cite{klus2020data}. gEDMD approximates a Koopman semigroup generator via Galerkin method. We solve the problem~\eqref{eq:eigenproblem} with a hybrid search over \((\alpha,\beta)\) using a standard line search method and gradient descent with iterative eigen-decompositions, which is a modification of the numerical method from \cite{cho2025koopman}.
\subsection{Nonlinear system under coarse sampling}
Consider the nonlinear system from \cite{klus2020data}
\begin{align}
& \dot{x}_1(t)=\gamma x_1(t), \\
& \dot{x}_2(t)=\delta\left(x_2(t)-x_1^2(t)\right),
\end{align} where the eigenfunction capturing the nonlinear effect is $\eig(x_1(t),x_2(t))=\frac{2 \gamma-\delta}{\delta} x_2(t)+x_1^2(t)$ with associated eigenvalue $\lambda= \delta$. For the experiment, we set \( (\delta,\gamma) = (-0.7, -0.8)\) and use monomials up to order 4 without a constant term to reconstruct Koopman eigenfunctions. We generate 20 different trajectories from uniformly distributed initial states in \([-2,2] \times [-2,2]\). Then we sample data points from trajectories obtained from numerical integration with different time intervals \(\{ 0.05, 0.25, 0.5\}\). Initial guesses \((\alpha,\beta)\) are chosen from \([-3, 1] \times [-1,1]\). The learned spectrum is shown in \Cref{fig:spectrum_klus}. 
\begin{figure}[bt!]
    \centering
    \includegraphics[width=0.85\linewidth]{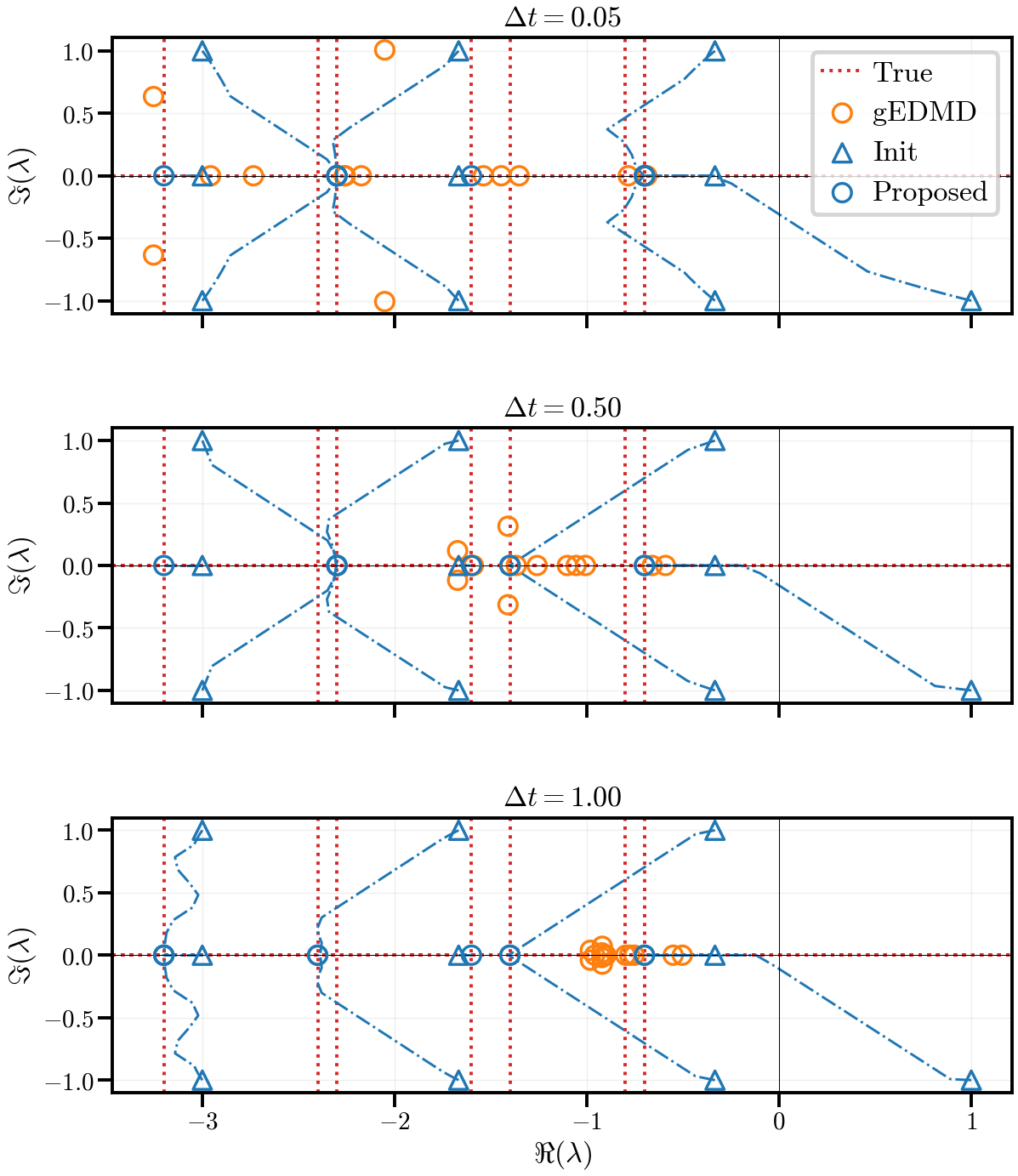}
    \caption{Spectrum obtained from gEDMD and the proposed method. Blue triangles denote the location of initial guesses of the Koopman eigenvalue. Dotted blue lines show trajectories of the eigenvalue update. Orange and blue circles denote the learned Koopman eigenvalue from gEDMD and the proposed formulation respectively. The red dotted lines are true eigenvalues.}
    \label{fig:spectrum_klus}
\end{figure}
Across all time intervals, the proposed method learns the true Koopman eigenvalues for all initial guesses, while gEDMD partially recovers the true eigenvalues when \(\interval{}=0.05\). gEDMD also exhibits spurious eigenvalues across all tested time intervals. All learned eigenvalues of the proposed method are combinations \(n \cdot \gamma + m \cdot \delta\) for some positive integers \(n, m\) (not both zero), and hence valid Koopman eigenvalues.
\addtolength{\textheight}{0cm}   

\subsection{High frequency dynamics and the effect of sampling}
We test the proposed formulation with the harmonic oscillator:
\begin{align}
    \dot{x}_1(t) &= -\omega x_2(t)\\
    \dot{x}_2(t) &= \omega x_1(t).
\end{align} The true Koopman eigenvalues are \(i\omega\) and \(-i\omega\), and the system has a conserved Hamiltonian \(0.5 (x_1^2(t) + x_2^2(t))\) at \(\lambda=0\).
We let \(\omega=50\) and set \(\interval{}=\{0.01, 0.2\}\). We test the formulation under both regular and irregular sampling: over a time horizon of length 4, a time interval of \(\delta = 0.2\) gives 20 observations while \(\delta = 0.01\) yields 400 observations. Then we randomly sample 20 observations from the 400 finely-sampled observations to test whether irregular sampling breaks aliasing. Initial Koopman eigenvalues \((\alpha,\beta)\) are chosen from \([-2, 2] \times [-60,60]\). 

The loss landscape of \eqref{eq:eigenproblem} in \(\beta\) is shown in \Cref{fig:landscape_harmonic}.
\begin{figure}[bt!]
    \centering
    \includegraphics[width=0.8\linewidth]{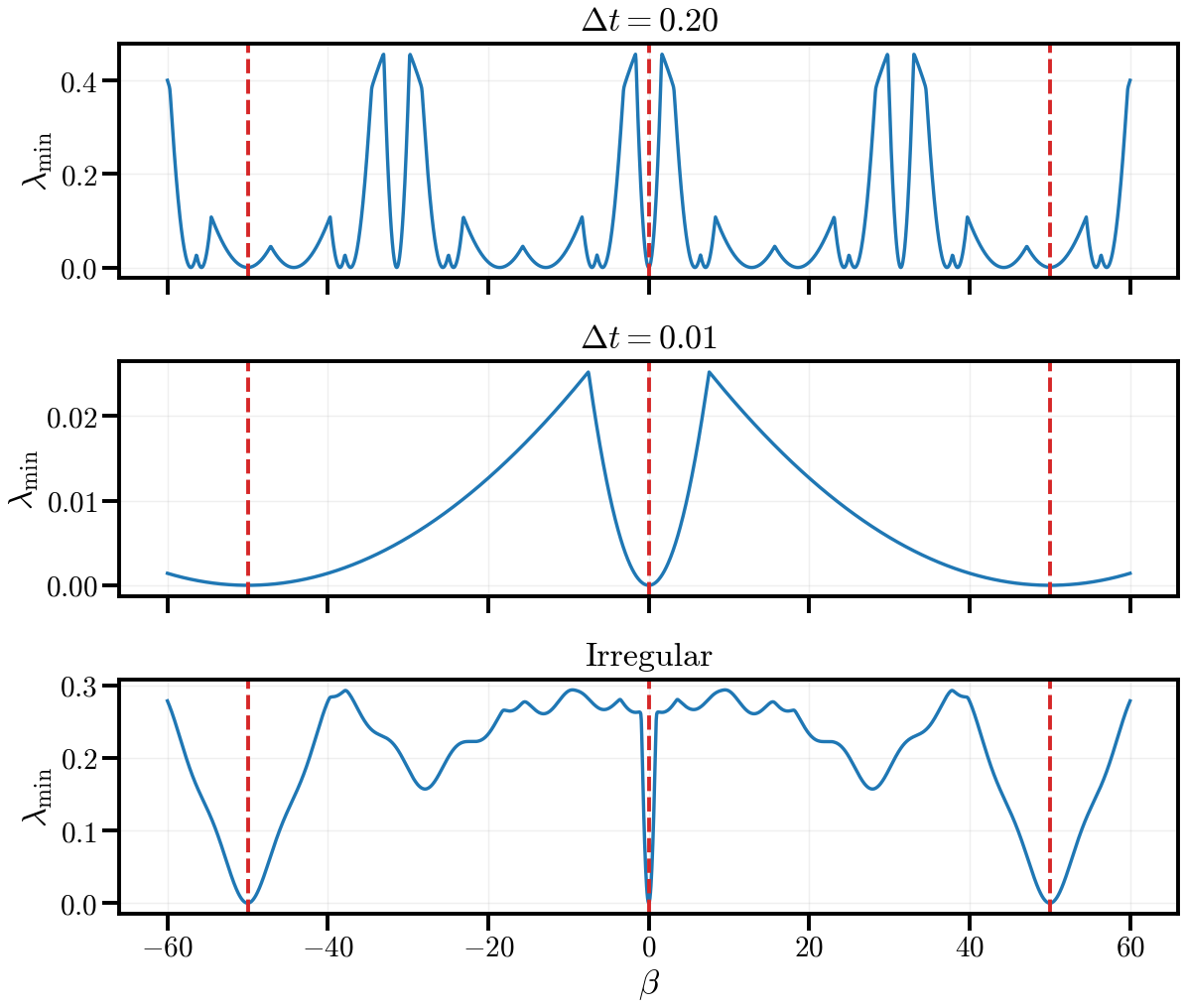}
    \caption{Loss landscape of \eqref{eq:eigenproblem}. The red vertical dotted lines show the location of true Koopman eigenvalues. The periodicity structure varies with the sampling pattern.}
    \label{fig:landscape_harmonic}
\end{figure}
For both cases of regular sampling, we observe periodicity and symmetry at \(\beta=0\). However, for irregular sampling, only symmetry is observed. Consistent with our analysis, irregularity makes true Koopman frequencies more identifiable for numerical methods using first-order information. The learned spectrum reflects these properties, as shown in \Cref{fig:spectrum_harmonic}. The proposed method is able to drive the real part of the eigenvalue to zero. However, due to the loss geometry, under regular sampling, the learned eigenvalues remain far from the true frequency, and the failure modes vary. As depicted in \(\interval{}=0.2\), the method can converge to a different aliased frequency. Also the method may halt because of plateaus near the true frequency as in \(\interval{}=0.01\); the gradient magnitude becomes too small. We observe that the failure modes can be avoided by relying solely on the line search method, though at additional computational cost. A detailed comparison is omitted due to space constraints. Under irregular sampling, the proposed method successfully learns the true Koopman eigenvalues, except for a few initial guesses. gEDMD captures the Hamiltonian structure but suffers from spurious eigenvalues even using small regular time intervals.

\begin{figure}[bt!]
    \centering
    \includegraphics[width=0.85\linewidth]{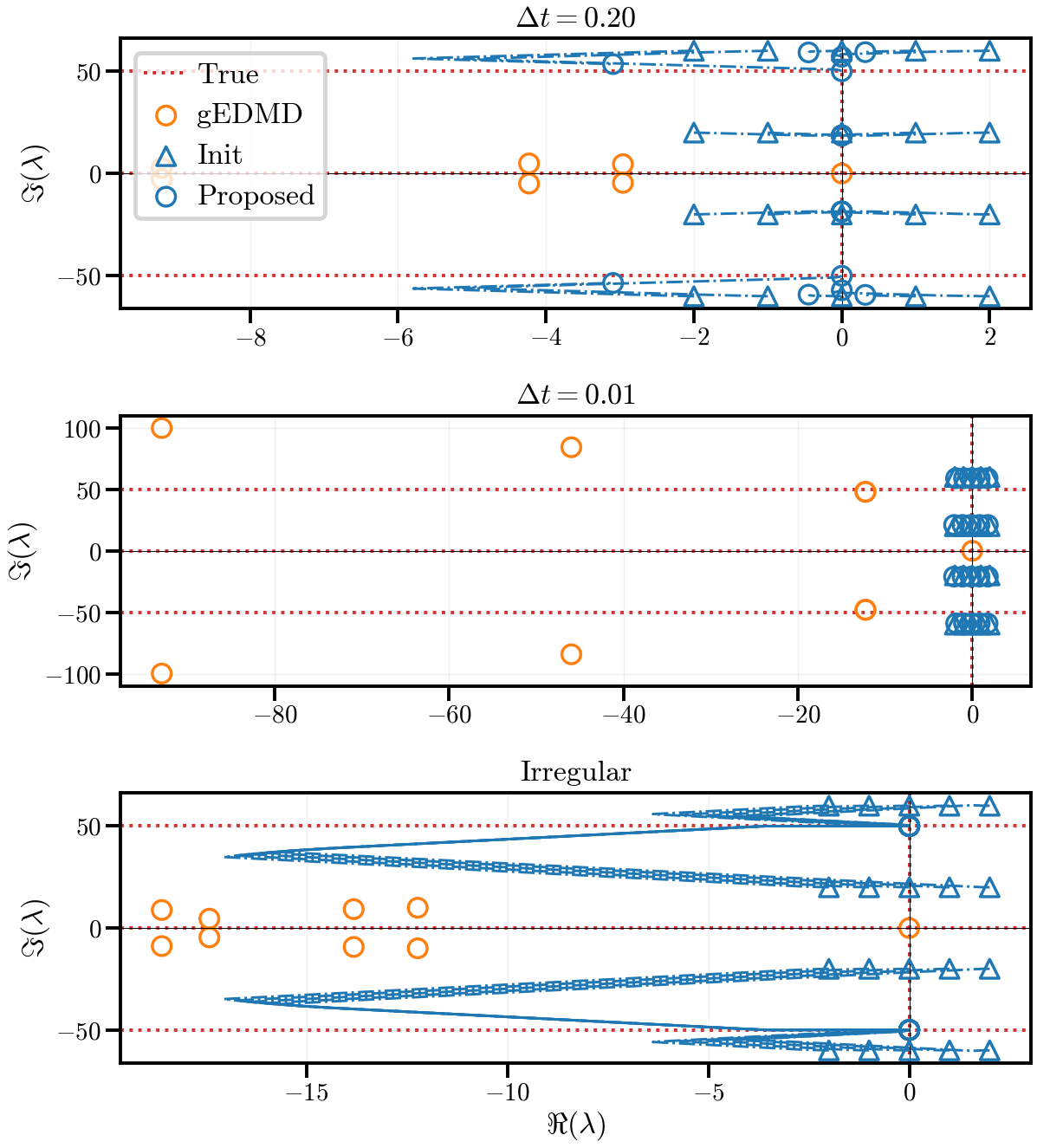}
    \caption{Spectrum obtained from gEDMD and the proposed method. The figure contains the same type of information as \Cref{fig:spectrum_klus}.}
    \label{fig:spectrum_harmonic}
\end{figure}

\section{Conclusions}
We formulate the problem of learning Koopman eigenfunctions from arbitrarily sampled observations as an optimization problem, extending \cite{cho2025koopman}. The analysis reveals aliasing that reflects the identifiability limit, and also phase alignment effects. We test the formulation on canonical systems and show that it is robust, with greater precision than gEDMD in reconstructing the Koopman spectrum. The numerical results also support the idea that irregular sampling can improve the effectiveness of numerical methods. We expect that our results can inform optimization of Koopman-based methods, and broaden the applicability of Koopman operator theory to complex data. 

\bibliography{IEEEabrv,reference} 

@article{koopman1931hamiltonian,
  title={Hamiltonian systems and transformation in {H}ilbert space},
  author={Koopman, Bernard O},
  journal={Proceedings of the National Academy of Sciences of the United States of America},
  volume={17},
  number={5},
  pages={315},
  year={1931},
  publisher={National Academy of Sciences},
  doi={10.1073/pnas.17.5.315},
  used={yes}
}

@article{brunton2016extracting,
  title={Extracting spatial--temporal coherent patterns in large-scale neural recordings using dynamic mode decomposition},
  author={Brunton, Bingni W and Johnson, Lise A and Ojemann, Jeffrey G and Kutz, J Nathan},
  journal={Journal of neuroscience methods},
  volume={258},
  pages={1--15},
  year={2016},
  publisher={Elsevier},
  doi={10.1016/j.jneumeth.2015.10.010},
    used={}
}

@article{mezic2013analysis,
  title={Analysis of fluid flows via spectral properties of the {K}oopman operator},
  author={Mezi{\'c}, Igor},
  journal={Annual Review of Fluid Mechanics},
  volume={45},
  pages={357--378},
  year={2013},
  publisher={Annual Reviews},
  doi={10.1146/annurev-fluid-011212-140652},
  used={}
}

@article{kou2022data,
  title={Data-driven eigensolution analysis based on a spatio-temporal {K}oopman decomposition, with applications to high-order methods},
  author={Kou, Jiaqing and Le Clainche, Soledad and Ferrer, Esteban},
  journal={Journal of Computational Physics},
  volume={449},
  pages={110798},
  year={2022},
  publisher={Elsevier},
  doi={10.1016/j.jcp.2021.110798},
  used={}
}

@article{avila2020data,
  title={Data-driven analysis and forecasting of highway traffic dynamics},
  author={Avila, AM and Mezi{\'c}, I},
  journal={Nature communications},
  volume={11},
  number={1},
  pages={1--16},
  year={2020},
  publisher={Nature Publishing Group},
  doi={10.1038/s41467-020-15582-5},
  used={}
}

@article{mezic2020spectrum,
  title={Spectrum of the {K}oopman operator, spectral expansions in functional spaces, and state-space geometry},
  author={Mezi{\'c}, Igor},
  journal={Journal of Nonlinear Science},
  volume={30},
  number={5},
  pages={2091--2145},
  year={2020},
  publisher={Springer},
  doi={https://doi.org/10.1007/s00332-019-09598-5}
}

@article{bevanda2021koopman,
  title={{K}oopman operator dynamical models: Learning, analysis and control},
  author={Bevanda, Petar and Sosnowski, Stefan and Hirche, Sandra},
  journal={Annual Reviews in Control},
  volume={52},
  pages={197--212},
  year={2021},
  publisher={Elsevier},
  doi={10.1016/j.arcontrol.2021.09.002},
}

@inproceedings{folkestad2020extended,
  title={Extended dynamic mode decomposition with learned {K}oopman eigenfunctions for prediction and control},
  author={Folkestad, Carl and Pastor, Daniel and Mezic, Igor and Mohr, Ryan and Fonoberova, Maria and Burdick, Joel},
  booktitle={2020 american control conference (acc)},
  pages={3906--3913},
  year={2020},
  organization={IEEE},
  doi={10.23919/ACC45564.2020.9147729},
}

@article{korda2020optimal,
  title={Optimal construction of {K}oopman eigenfunctions for prediction and control},
  author={Korda, Milan and Mezi{\'c}, Igor},
  journal={IEEE Transactions on Automatic Control},
  volume={65},
  number={12},
  pages={5114--5129},
  year={2020},
  publisher={IEEE},
  doi={10.1109/TAC.2020.2978039},
}

@article{kaiser2021data,
  title={Data-driven discovery of {K}oopman eigenfunctions for control},
  author={Kaiser, Eurika and Kutz, J Nathan and Brunton, Steven L},
  journal={Machine Learning: Science and Technology},
  volume={2},
  number={3},
  pages={035023},
  year={2021},
  publisher={IOP Publishing},
  doi={10.1088/2632-2153/abf0f5},
}

@inproceedings{deka2023path,
  title={Path-integral formula for computing {K}oopman eigenfunctions},
  author={Deka, Shankar A and Narayanan, Sriram SKS and Vaidya, Umesh},
  booktitle={2023 62nd IEEE Conference on Decision and Control (CDC)},
  pages={6641--6646},
  year={2023},
  organization={IEEE},
  doi={10.1109/CDC49753.2023.10384288}
}

@article{cho2025koopman,
  title={{K}oopman representations with irregular time intervals},
  author={Cho, Younghwan and Sowers, Richard},
  journal={Physica D: Nonlinear Phenomena},
  pages={135062},
  year={2025},
  publisher={Elsevier},
  doi={https://doi.org/10.1016/j.physd.2025.135062}
}

@article{williams2015edmd,
  title={A data--driven approximation of the koopman operator: Extending dynamic mode decomposition},
  author={Williams, Matthew O and Kevrekidis, Ioannis G and Rowley, Clarence W},
  journal={Journal of Nonlinear Science},
  volume={25},
  pages={1307--1346},
  year={2015},
  publisher={Springer},
  doi={10.1007/s00332-015-9258-5},
}

@article{klus2020data,
  title={Data-driven approximation of the {K}oopman generator: Model reduction, system identification, and control},
  author={Klus, Stefan and N{\"u}ske, Feliks and Peitz, Sebastian and Niemann, Jan-Hendrik and Clementi, Cecilia and Sch{\"u}tte, Christof},
  journal={Physica D: Nonlinear Phenomena},
  volume={406},
  pages={132416},
  year={2020},
  publisher={Elsevier},
  doi={10.1016/j.physd.2020.132416}
}

@article{kaur2023deep,
  title={Deep learning for multiple sclerosis differentiation using multi-stride dynamics in gait},
  author={Kaur, Rachneet and Levy, Joshua and Motl, Robert W and Sowers, Richard and Hernandez, Manuel E},
  journal={IEEE Transactions on Biomedical Engineering},
  volume={70},
  number={7},
  pages={2181--2192},
  year={2023},
  publisher={IEEE},
  doi={10.1109/TBME.2023.3238680}
}

@article{martin2006gait,
  title={Gait and balance impairment in early multiple sclerosis in the absence of clinical disability},
  author={Martin, Clarissa L and Phillips, Beverley A and Kilpatrick, TJ and Butzkueven, Helmut and Tubridy, Niall and McDonald, Elizabeth and Galea, MP},
  journal={Multiple Sclerosis Journal},
  volume={12},
  number={5},
  pages={620--628},
  year={2006},
  publisher={Sage Publications Sage CA: Thousand Oaks, CA},
  doi={https://doi.org/10.1177/13524585060706}
}

@article{fujito2006effect,
  title={Effect of sensor spacing on performance measure calculations},
  author={Fujito, Iris and Margiotta, Rich and Huang, Weimin and Perez, William A},
  journal={Transportation research record},
  volume={1945},
  number={1},
  pages={1--11},
  year={2006},
  publisher={SAGE Publications Sage CA: Los Angeles, CA},
  doi={https://doi.org/10.1177/036119810619450010}
}

@article{mezic2004comparison,
  title={Comparison of systems with complex behavior},
  author={Mezi{\'c}, Igor and Banaszuk, Andrzej},
  journal={Physica D: Nonlinear Phenomena},
  volume={197},
  number={1-2},
  pages={101--133},
  year={2004},
  publisher={Elsevier},
  doi={10.1016/j.physd.2004.06.015},
}

@inproceedings{budivsic2009approximate,
  title={An approximate parametrization of the ergodic partition using time averaged observables},
  author={Budi{\v{s}}i{\'c}, Marko and Mezi{\'c}, Igor},
  booktitle={Proceedings of the 48h IEEE Conference on Decision and Control (CDC) held jointly with 2009 28th Chinese Control Conference},
  pages={3162--3168},
  year={2009},
  organization={IEEE},
  doi={10.1109/CDC.2009.5400512},
}

@book{tao2023topics,
  title={Topics in random matrix theory},
  author={Tao, Terence},
  volume={132},
  year={2023},
  publisher={American Mathematical Society}
}

@book{rellich1969perturbation,
  title={Perturbation theory of eigenvalue problems},
  author={Rellich, Franz},
  year={1969},
  publisher={CRC Press}
}

@article{shannon2006communication,
  title={Communication in the presence of noise},
  author={Shannon, Claude E},
  journal={Proceedings of the IRE},
  volume={37},
  number={1},
  pages={10--21},
  year={2006},
  publisher={IEEE}
}

@article{candes2008introduction,
  title={An introduction to compressive sampling},
  author={Cand{\`e}s, Emmanuel J and Wakin, Michael B},
  journal={IEEE signal processing magazine},
  volume={25},
  number={2},
  pages={21--30},
  year={2008},
  publisher={IEEE},
  doi={10.1109/MSP.2007.914731}
}






\end{document}